\documentclass[11pt]{article}
\usepackage[utf8]{inputenc}

\usepackage[linktocpage]{hyperref}
\usepackage[dvipsnames]{xcolor}

\usepackage{amsmath}
\usepackage{amssymb}
\usepackage{amsthm}
\usepackage{amsfonts}
\usepackage{graphicx} 
\usepackage{wrapfig}
\usepackage{xspace}
\usepackage{thm-restate}
\usepackage{cleveref}

\newcommand{\poly}{\operatorname{\textsc{{\rm poly}}}}
\newcommand{\upd}{\operatorname{\textsc{update}}}
\newcommand{\parnt}{\operatorname{\textsc{parent}}}

\usepackage[style=trad-alpha,maxcitenames=2,backend=biber]{biblatex}
\newcommand{\citet}[1]{\citeauthor{#1}\xspace\cite{#1}}
\newcounter{mnote}[section]

\definecolor{ForestGreen}{rgb}{0.0333,0.4451,0.0333}
\definecolor{DarkRed}{rgb}{0.65,0,0}
\definecolor{Red}{rgb}{1,0,0}
\hypersetup{
    unicode=false,          
    colorlinks=true,        
    linkcolor=BrickRed,          
    citecolor=OliveGreen,        
    filecolor=magenta,      
    urlcolor=cyan           
}

\usepackage{float}
\usepackage[dvipsnames]{xcolor}
\usepackage{tikz}
\usepackage{subcaption}
\usepackage{wrapfig}
\usepackage{algorithm}
\usepackage[noend]{algpseudocode}
\usepackage{fullpage}

\newtheorem{definition}{Definition}

\newtheorem{theorem}{Theorem}[section]
\newtheorem{lemma}[theorem]{Lemma}
\newenvironment{claim}[1]{\par\noindent\textbf{Claim.}\space#1}{}
\newenvironment{claimproof}[1]{\par\noindent\emph{Proof.}\space#1}{\hfill $\blacksquare$}

\usepackage{tikz}

\let\citet\cite

\begin{document}

\author{
        \begin{tabular}[t]{c@{\extracolsep{5em}}c} \centering
                 Niklas Dahlmeier  & \qquad D Ellis Hershkowitz\thanks{Supported by NSF grant CCF-2403236.} \\
                 \small RWTH Aachen   & \small \qquad Brown University \\
                 \small \texttt{dahlmeier@gdm.rwth-aachen.de}  &\small \qquad \texttt{delhersh@gmail.com}
        \end{tabular}
}


\title{Low Recourse Arborescence Forests\\ Under Uniformly Random Arcs}
\date{}

\maketitle

\begin{abstract}
    In this work, we study how to maintain a forest of arborescences of maximum arc cardinality under arc insertions while minimizing recourse---the total number of arcs changed in the maintained solution. This problem is the ``arborescence version'' of max cardinality matching. 
    
    On the impossibility side, we observe that even in this insertion-only model, it is possible for $m$ adversarial arc arrivals to necessarily incur $\Omega(m \cdot n)$ recourse, matching a trivial upper bound of $O(m \cdot n)$. On the possibility side, we give an algorithm with expected recourse $O(m \cdot \log^2 n)$ if all $m$ arcs arrive uniformly at random. 
\end{abstract}  
\thispagestyle{empty}
\newpage
\setcounter{page}{1}

\section{Introduction}
\label{sec:Introduction}

Arborescences (a.k.a.\ directed spanning trees) are one of the most well-studied objects in algorithmic graph theory. Specifically, given a digraph $G = (V,A)$ an arborescence is a subgraph $T \subseteq A$ where every vertex has in-degree at most $1$ and $T$ is a tree if we forget arc directions. Since the formulation of poly-time algorithms for computing a minimum weight arborescence \cite{chu1965shortest,edmonds1967optimum}, arborescences have served as an important case study across areas of algorithms and combinatorial optimization. These include near-linear time, primal-dual, randomized, and approximation algorithms, as well as integral polyhedra / totally unimodular (TU) matrices, edge splitting and tree packing theorems \cite{berczi2009packing,gabow1986efficient,drescher2010approximation,gabow1991matroid,korte2018spanning,bhalgat2008fast,laekhanukit2012rounding}.

However, one notable area in which arborescences are poorly understood is the dynamic setting. In dynamic settings, our goal is to maintain a solution as the input changes over discrete time steps. A well-studied notion of the quality of a dynamic algorithm is that of recourse: the total number of changes made to the solution over time. Low-recourse algorithms are desirable both because recourse lower bounds the time to update a solution and because solutions that do not change very much over time are often desirable in their own right. For these reasons, many classic algorithmic problems have been studied in the low recourse setting, including set cover \cite{gupta2017online,gupta2020fully}, matching \cite{grove1995online,chaudhuri2009online,bosek2014online,bernstein2019online,megow2020online,angelopoulos2020online}, load balancing \cite{gupta2014maintaining,krishnaswamy2023online}, minimum spanning and Steiner tree \cite{imase1991dynamic,gupta2014online,megow2016power,gu2013power}, TSP \cite{megow2016power}, clustering \cite{fichtenberger2021consistent,lattanzi2017consistent,cohen2019fully,lkacki2024fully} and general convex body chasing \cite{argue2021chasing,sellke2023chasing,bansa2018nested,buchbinder2009online,bhattacharya2023chasing}. 

Very recently, Buchbinder et al. \citet{buchbinder2024maintaining} gave low-recourse algorithms for the dynamic version of what we will call the maximum arborescence forest problem.\footnote{More generally, they give results for low recourse matroid intersection where one matroid is a partition matroid and each part of the partition is incrementally revealed.}
\begin{definition}[(Maximum) Arborescence Forest]
    Given a digraph $G = (V, A)$, we say a subgraph $F \subseteq A$ is an arborescence forest if each (weakly) connected component of $F$ is an arborescence. Such a subgraph is maximum if it has maximum arc cardinality among all arborescence forests.
\end{definition}
\noindent We refer to any vertex of an arborescence forest with in-degree $0$ as a \emph{root}. See \Cref{branching example}. The dynamic version of this problem which we study in this work is as follows.

\begin{definition}[Incremental Maximum Arborescence Forest]\label{dfn:incMaxArb}
    In the incremental maximum arborescence forest, we are given a vertex set $V$ and a sequence of arcs $a_1, a_2, \ldots, a_m$ over time steps. In the $i$-th time step we must output a maximum arborescence forest $F^{(i)}$ for digraph $G^{(i)} := (V, \bigcup_{j \leq i} \{ a_j\} )$. Our goal is to minimize recourse, defined as\footnote{We use arc deletions and recourse interchangeably. Recourse is often instead defined as a symmetric difference between solutions. The structure of maximum arborescence forests and the fact that our solutions monotonically increase in size up to $n-1$ means that these two quantities are the same up to a multiplicative $2$ and additive $n-1$.}
    \begin{align*}
        \sum_{i=1}^{m-1} |F^{(i)}\setminus F^{(i+1)}|.
    \end{align*}
\end{definition}
\noindent Buchbinder et al. \citet{buchbinder2024maintaining} showed that $O(n \log ^2 n)$ recourse is possible if $n$ vertices---rather than edges---arrive adversarially. That is, if in each time step all in-arcs of a unique vertex arrive, then $O(n \log ^2 n)$ recourse is possible after $n$ vertex arrivals. On the other hand, to date nothing is known for this problem if arcs, rather than vertices, arrive over time.

\begin{figure}[ht]
    \centering
    \includegraphics[scale=0.75]{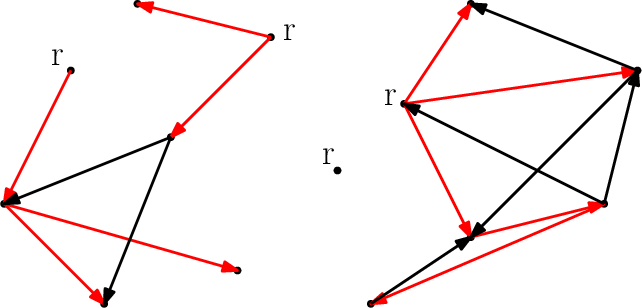}
\caption{Red: an arborescence forest where each root is labeled ``$r$''. Black: all other arcs.}
    \label{branching example}
\end{figure}

\subsection{Our Results}
We give results for incremental maximum arborescence forest when arcs arrive over time.

On the hardness side, we observe that if $m$ arcs arrive non-adaptively but adversarially, then it is impossible to achieve recourse $o(m \cdot n)$.
\begin{restatable}{theorem}{lowerTheorem}\label{thm:lower}
    Given integral $n \geq 0$, there exists a fixed $n$ vertex instance of incremental maximum arborescence forest with $O(n)$ arc insertion such that every solution has recourse at least $\Omega(n^2)$.
\end{restatable}
\noindent Note that it is always possible to achieve recourse $O(n)$ per each arc insertion (as one can fully change the maintained solution with $O(n)$ recourse). Thus, the above impossibility result matches a trivial upper bound of $O(m \cdot n)$. (In \Cref{subsection: adversarial case} 
we observe that a similar lower bound holds for the problem of dynamically maintaining a minimum-cost arborescence, even when all edges have cost $0$ or $1$ and parallel arcs are prohibited.)

Since adversarial arc arrivals prohibit nontrivial upper bounds, it becomes natural to study random arc arrivals. On the possibility side, we demonstrate that if $m$ arcs arrive uniformly at random chosen from the complete graph, then expected recourse $O(m \log^2 n)$ is always possible.
\begin{restatable}{theorem}{upperTheorem}\label{thm:upper}
    There is a polynomial-time algorithm for incremental maximum arborescence forest with expected recourse $O(m \log ^ 2 n)$ if all $m$ arcs arrive uniformly at random.
\end{restatable}
\noindent Similar models where updates happen uniformly at random have been studied by other works in dynamic algorithms. See, for example, work on densest subgraph \cite{epasto2015efficient} and subgraph counting \cite{henzinger2022complexity}.

Our results can be seen as a natural complement to existing work in maximum cardinality (bipartite) matching. Specifically, maximum arborescence forest is, in some sense, to arborescence problems as maximum cardinality matching is to matching problems: like maximum arborescence forest, maximum bipartite matching is a special case of finding the largest common independent set of two matroids.\footnote{For maximum matching this is the intersection of two partition matroids whereas for maximum arborescence forest this is the intersection of a partition matroid and the graphic matroid.} Furthermore, like maximum arborescence forest, in maximum bipartite matching $O(n \log ^2 n)$ recourse is possible in adversarial vertex arrival models but, similar to our result, $\Omega(n^2)$ recourse is necessary in adversarial edge arrival models \cite{bernstein2019online}. Furthermore, motivated similarly to us, Bernstein and Dudeja \citet{bernstein2020online} studied random edge arrival models for low-recourse bipartite matching, showing that if the graph is chosen adversarially but $n$ edges arrive in a uniformly random order, expected recourse of $O(n \log n)$ and $O(n \log ^ 2 n)$ on paths and trees respectively is possible.\footnote{This work also shows that recourse $\Omega\left(\frac{n^2}{\log n}\right)$ is necessary for general graphs in this adversarial graph but random edge arrival setting.} We show that similar results are possible for maximum arborescence forest on \emph{general graphs} when arcs arrive uniformly at random (among all possible arcs, not an adversarially-chosen subset).

\subsection{Challenges and Techniques}
Before giving a formal description of our results, we give an overview of challenges and techniques.

\subsubsection{Adversarial Recourse Lower Bound}
Our recourse lower bound of $\Omega(n^2)$ over the course of $O(n)$ arc insertions follows from a simple ``bidirected'' path example.

The basic idea is to insert arcs pointing in both directions on a path so that there is always a unique maximum arborescence forest whose nontrivial arborescence either goes from left to right or right to left in the path. We then repeatedly add arcs to the end of this nontrivial arborescence to flip the direction that it points so that once $\Omega(n)$ arcs have been inserted each new arc inserted incurs recourse $\Omega(n)$.

\subsubsection{Randomized Recourse Upper Bound}
Most of our work focuses on showing our upper bound in the uniformly random arc arrival setting. 

\paragraph{Basic Idea of the Algorithm.}
We begin by discussing the intuition for our algorithm. It is not too hard to see that a forest of arborescences $F$ is maximum in a graph $G$ if and only if there do not exist distinct roots $r$ and $r'$ of $F$ with a path from $r$ to $r'$ in $G$ (\Cref{onlyroots}).

One natural algorithmic idea, then, is to repeatedly reconfigure one's solution to guarantee that no such paths exist. It is not too hard to see that adding any such path to a solution and appropriately deleting parent arcs of vertices in such a path results in a new arborescence forest of size one larger than the previous (\Cref{lem:correctness}). Furthermore, the result of reconfiguring along such a path from $r$ to $r'$ is that $r$ remains a root but $r'$ is no longer a root; all of the vertices that used to have $r'$ as the root of their arborescence now have $r$ as their root.

\paragraph{A Weak Recourse Bound.}
However, since our goal is ultimately a low recourse algorithm, we must guarantee that we do not need to perform too many such reconfigurations. Likewise, we must guarantee that the paths between roots along which we update are not too long (and therefore do not result in too many parent arc deletions). 

Regarding the former point, one can see that if we already had a maximum arborescence forest, we need only perform at most one such reconfiguration when a new arc arrives. Intuitively, this is because adding a single arc to a graph increases its maximum arborescence forest cardinality by at most $1$ (\Cref{lem:correctness}). 

Regarding the latter point, we can show that if there exists some path between two roots, then there always exists a nicely structured one (which we call a ``feasible path''; see \Cref{dfn:feasiblePath}) between some other roots $r$ and $r'$. These nicely structured paths satisfy the property that, if we update along them, the maximum number of arc deletions we incur is at most the size of the arborescence of $r'$ in the current arborescence forest. Of course, in general, the size of the arborescence of $r'$ may be as large as $\Omega(n)$. Nonetheless, we will leverage this seemingly weak bound on the arborescence to argue $O(\log ^ 2 n)$ amortized recourse.

The key to leveraging the above weak recourse bound will be to observe that a given vertex $v$ must delete its parent arc in our solution only when it is contained in the arborescence of a root $r'$ where we just added a new arc incident to the in-component of $r'$.\footnote{The in-component of a vertex $v$ in a digraph comprises all vertices that can reach $v$. The out-component of $v$ comprises all the vertices that $v$ can reach and the strongly connected component of $v$ is the intersection of its in- and out-components.} Thus, if the in-component of $r'$ is very small, we do not expect $v$ to be forced to delete its parent arc very often.

Of course, in general, the in-component of $r'$ can have size $\Omega(n)$ and so it will not \emph{always} be the case that $v$ is contained in an arborescence whose root has a small in-component.

\paragraph{Using Random Graph Structure.}
The last piece of our strategy is an amortization argument which shows that, by properties of our random model and the structural properties of our paths, \emph{on average} vertices tend to be in arborescences whose in-components are small and so they need not delete their parent arcs too often.

The first step towards arguing this is to observe that our random model can be equivalently understood in a way that allows us to inherit some structural results from $D(n,p)$, the random directed $n$-node graph where each arc is independently present with probability $p$. In particular, if we assign to each arc a uniformly random value in $[0,1]$ and then sequentially add the $m$ arcs with lowest value (in ascending order of value) then this is identical to our model. This then allows us to note that if we take all arcs whose value is less than some $p$, then our graph at this point is distributed identically to $D(n,p)$ and so we can make use of structural theorems known for $D(n,p)$. 

Likewise, we can argue that a desirable property is always exhibited by our graph by either (1) noting that the property is monotone and so proving it for $D(n,p)$ proves it for $p' < p$ or (2) subdividing $[0,1]$ into small intervals each of which contains at most one arc, proving the property for each interval with high probability, and then union bounding over all intervals.

Notably, for $p$ much larger than $1/n$ it is known that $D(n,p)$ has no in-components of size in $[\Omega(\log n), O(n)]$ for a suitable hidden constant (\Cref{lemma: at most log n helper}) with high probability. Furthermore, our algorithm guarantees:
\begin{itemize}
    \item The in-component of a root is always fully contained in its arborescence (\Cref{lem:inCompInArb}) and;
    \item Merging the arborescence of $r'$ into that of $r$ does not change the in-component of $r$ (\Cref{lem:presInComp}).
\end{itemize}
Thus, even if a vertex $v$ is in an arborescence whose root has an in-component of size $\Omega(n)$, it can only be merged into $O(1)$ such arborescences until only arborescences with in-components of size $O(\log n)$ remain (since otherwise its arborescence would have size $> n$). Once $v$ is merged into such an arborescence whose root's in-component has size $O(\log n)$, $v$ will contribute essentially nothing in recourse for the subsequent $\Theta(n / \log n)$ rounds (arc insertions).

More precisely, we divide our recourse analysis into two phases, which correspond to $p < 2/n$ and $p > 2/n$. 
The two phases use somewhat sophisticated amortization arguments similar to the above to show amortized recourse $O(\log n)$ and $O(\log ^ 2 n)$ respectively.


\subsubsection{Discussion of Techniques}
While our results make use of random graph structure, the key insight of our work is not necessarily only applicable in uniformly random graphs. In particular, the idea that structuring one's solution so that, on average, vertices tend to be in arborescences whose roots have small in-components might feasibly be useful for other randomized models. 

For instance, our techniques may prove useful for the above-mentioned model in which the graph is fixed adversarially but arcs arrive in a random order. Notably, it is possible to apply the above insight to our lower bound instance to achieve low recourse in this random model for this particular graph. Whether this is possible in general graphs is an exciting open question for future work. However, some tempering of expectations is likely necessary: low-recourse algorithms have been proven impossible in this model for the closely-related maximum matching problem \cite{bernstein2020online} and so such results may also ultimately be impossible for arborescences.






\section{Lower Bound for Adversarial Arc Arrivals} \label{section: dynamic graphs}

In this section, we briefly explain a simple example that demonstrates that if arcs arrive adversarially, then, in general, recourse $\Omega(m \cdot n)$ is necessary after $m$ arc arrivals for maximum arborescence forest. Notably, this example is non-adaptive in the sense that the adversary can fix all arc insertions ahead of time without knowing what maximum arborescence forests the algorithm outputs. The idea is to build a bi-directed path of length $\Theta(n)$ such that the root of this path is always on one of the two ends. Then a new incoming arc is added to the other side of the path such that all arcs need to get flipped. This newly added arc then gets bi-directed as well and this process continues until the graph is strongly connected. See \Cref{fig:adversarialarborescence} for an illustration of the construction.

\lowerTheorem*
\begin{proof}	
	We now provide additional details. Suppose for sake of simplicity that $n$ is even (the odd case is identical except the vertex $v_{n/2}$ is not defined). Our graph consists of vertices $v_1, v_2, \ldots, v_n$.
	
	The first inserted arc $a^{(1)}$ is from $v_{n/2}$ to vertex $v_{n/2 + 1}$. At this point the root of the arborescence containing $v_{n/2 + 1}$ is $v_{n/2}$. The second arc now ``bidirects'' the first, namely we add arc $a^{(2)} = (v_{n/2+1}, v_{n/2})$.The next arc forces a change in the maximum arborescence forest. Namely $a^{(3)}$ goes from vertex $v_{n/2 + 2}$ to vertex $v_{n/2 + 1}$. The addition of this arc adds vertices $v_{n/2}$ and $v_{n/2 + 1}$ to the arborescence rooted at vertex $v_{n/2 + 2}$. Again the next arc added will ``bidirect'' the just added arc, namely $a^{(4)}$ goes from $v_{n/2 + 1}$ to $v_{n/2 + 2}$. 
	
	This process of repeatedly adding and bidirecting arcs to each side of the non-singleton arborescence continues until the graph is strongly connected. Every second arc insertion forces a recourse of $\frac{m-1}{2}$ after the insertion of $m$ arcs which gives a total recourse of $\Omega( n^2)$ for $m = \Omega(n)$.
\end{proof}
\noindent We note that our example is somewhat reminiscent of an example of Bernstein et al. \citet{bernstein2019online} which shows similar recourse lower bounds for adversarial edge arrival models.
    \begin{figure}[H]
    \centering
    \includegraphics[height=6cm, width=5cm]{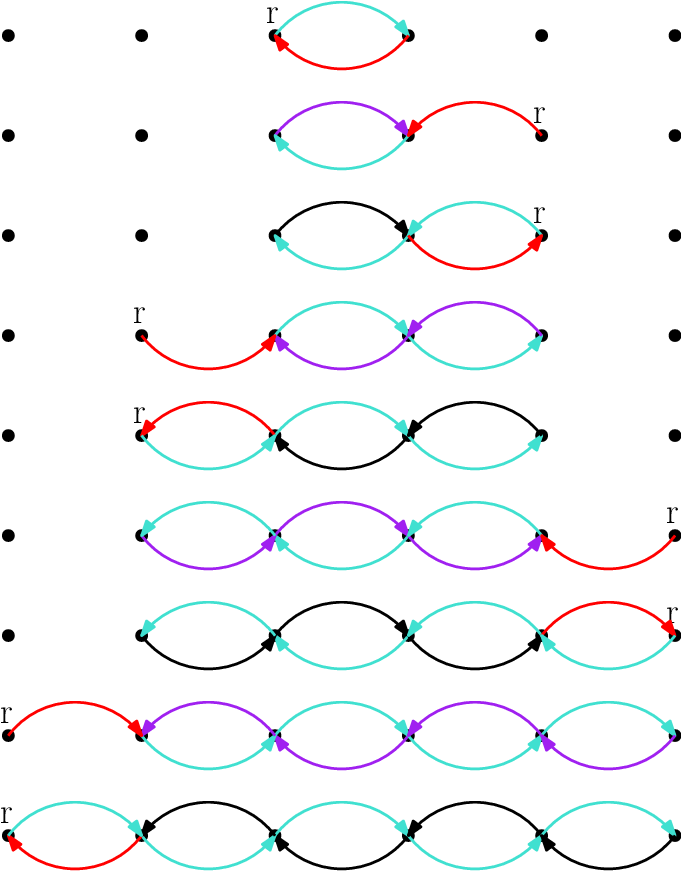} 
    \caption{Illustration for \Cref{thm:lower}. The $i$th row is $G^{(i)}$. Red arcs always denote the most recently added arc, turquoise arcs form the current arborescence forest and violet arcs are arcs deleted from the solution and hence count towards recourse.}
    \label{fig:adversarialarborescence}
\end{figure}

\section{Upper Bound for Uniformly Random Arc Arrivals}
\label{section: proof strategies}

In this section, we show how to achieve expected recourse $O(m \log ^ 2 n)$ over the course of $m$ arc insertions. Namely, we prove the following theorem.

\upperTheorem*

\subsection{Algorithm Description}
We begin by describing our algorithm for incremental maximum arborescence forest. Given an arborescence forest $F$ and a non-root vertex $v$, we let $\parnt(v)$ give the parent arc of $v$ in $F$ (i.e.\ the unique in-arc of $v$ in the arborescence which contains $v$ in $F$). To compute $F^{(i)}$ given $F^{(i-1)}$ we check if adding $a_{i}$ created \emph{some path} between two distinct roots of $F^{(i-1)}$. If doing so does, then we select a \emph{certain structured path} between roots, remove the parent arcs of vertices internal to this path and then add this path to $F^{(i-1)}$ to produce $F^{(i)}$. Note that there could be multiple roots that can reach a root after adding $a_{i}$ (though the destination root is the same for all) and we specify which one is chosen. 

More formally, the way we update along paths is captured by the following update.
\begin{definition}[Path Update]
    Let $F$ be an arborescence forest of graph $G$ with roots $r$ and $r'$ and let $P = (r, v_1, v_2, \ldots, v_k, r')$ be a path of $G$ from $r$ to $r'$. Then we define 
        \begin{align*}
            \upd(F, P) := \left(F \setminus \bigcup_i \parnt(v_i)\right) \cup P.
        \end{align*}
\end{definition}
The forms of the paths along which we update are given by the following notion of feasible paths. Roughly, these are paths that follow one arborescence, take the added arc and then stay inside a single strongly connected component. 
\begin{definition}[Feasible Paths]\label{dfn:feasiblePath}
    Let $F$ be an arborescence forest of graph $G$ containing two arborescences $T$ and $T'$ with respective roots $r$ and $r'$. Then we say that a path $P \subseteq G$ from $r$ to $r'$ is feasible if it is of the form $P = P_A \oplus P_{V}$ where the arcs of $P_A$ are contained in $T$ and the vertices of $P_V$ are contained in $T'$.
\end{definition}

Using the above definitions, we formally define our algorithm in pseudocode in \Cref{alg:main}.

\begin{algorithm}[H]
    \caption{Algorithm for Incremental Maximum Arborescence Forest}
    \label{alg:main}
    \begin{algorithmic}[0] 
            \State \textbf{Input:} Vertex set $V$ and a sequence of arcs $a_1, a_2, \ldots, a_m$ 
            \State \textbf{Output:} $F^{(1)}, F^{(2)}, \ldots, F^{(m)}$ s.t.\ each $F^{(i)}$ is a max arborescence forest for $G^{(i)}:= (V, \bigcup_{j \leq i} \{a_j\})$
            \State Let $F^{(0)} = (V, \emptyset)$
            \For{$i = 1, \ldots, m$}
                \If{$F^{(i-1)}$ has a feasible path $P$ in $G^{(i)}$ (\Cref{dfn:feasiblePath})}
                    \State $F^{(i)} \gets \upd\left(F^{(i-1)}, P\right)$
                    \Else
                    \State $F^{(i)} \gets F^{(i-1)} $
                \EndIf
            \EndFor 
        \Return $F^{(0)}, F^{(1)}, \ldots, F^{(m)}$
    \end{algorithmic}
\end{algorithm}

\subsection{Correctness of Algorithm}

We begin by proving that \Cref{alg:main} is correct; namely, we prove that $F^{(i)}$ is a maximum arborescence forest for $G^{(i)}$ for every $i$.

We begin by observing that our update results in a forest of arborescences.
\begin{lemma}\label{lem:updatePresForest}
    Let $G$ be a digraph and let $F \subseteq G$ be an arborescence forest with a path $P \subseteq G$ between distinct roots of $F$. Then $\upd(F, P)$ is a forest of arborescences.
\end{lemma}
\begin{proof}

Let $P = (r, v_1, v_2, \ldots, v_k, r')$ be the path and let $F' = \upd(F, P)$. First, observe that $\upd$ preserves the in-degree of all vertices, except for $r'$ which has its in-degree increased from $0$ to $1$. Since $F$ was a forest of arborescences, it follows that every vertex in $F'$ has in-degree at most $1$. Thus, to show that $F'$ is a forest of arborescences, it suffices to show that $F'$ contains no cycles if we forget about arc directions.

Suppose for the sake of contradiction that such a cycle existed and let $P'$ be a maximal contiguous subpath (arcs do not necessarily point in the same direction) of this cycle where all edges of $P'$ are contained in $F' \setminus P$ (such a subpath must exist since $P$ is itself a path). Let $u$ and $v$ be the first and last vertices of $P'$ with incident arcs $a_v$ and $a_u$ in $P'$. By the assumption that $a_v, a_u \in F' \setminus P$ we know that $a_v$ and $a_u$ have $v$ and $u$ as their tails (i.e.\ the arcs point \emph{away} from $u$ and $v$) since otherwise they would have been the parent arcs of $u$ and $v$ and would have been deleted by $\upd$. However, it follows that some vertex in $P'$ has in-degree at least $2$ in $F'$, contradicting the above argument that $F'$ has in-degree at most $1$. 
\end{proof}

We next state a simple helper lemma characterizing the maximum arborescence forests as those which have no paths between roots. This lemma is a special case of more general theorems known regarding matroid intersection---see e.g.\ \cite{cunningham1986improved}---but is considerably simpler in the special case of maximum arborescence forests and so we include a proof for completeness.
\begin{lemma}
\label{onlyroots}
    Given a digraph $G$, an arborescence forest $F \subseteq G$ is maximum if and only if there do not exist distinct roots $r$ and $r'$ of $F$ where $r$ has a path to $r'$ in $G$.
\end{lemma}
\begin{proof}
 First, suppose that we have a directed path $P$ from some root $r$ to some other root $r'$ and denote the arborescence starting at $r'$ by $T$. Now, simply backtrack this path $P$ starting at vertex $r'$ until we hit the first vertex that does not belong to $T$ and call that vertex $v$. Let $r''$ be the root of said vertex $v$. Now $r''$ can become the new root of the arborescence $T$ that started at $r'$ simply by adding the path from $v$ to $r'$ and deleting the unique incoming arcs of the arborescence $T$ to the vertices in the path $P$ from $v$ to $r'$. This shows that the arborescence forest was not maximum.

For the other direction, suppose the arborescence forest $F$ is not maximum, then there is a different arborescence forest $F'$ with strictly fewer arborescences that still covers all vertices. By the pigeonhole principle, there must be at least one arborescence $T$ in $F'$ that contains at least two roots $r$ and $r'$ from the original arborescence forest. Denote by $r''$ the root of this arborescence $T$. If there is a path from $r$ to $r'$ or vice versa we are done. Else we know that $r$ and $r'$ are not in the strong component of $r''$. We also know that there must be some root $s$ in the original arborescence forest $F$ whose arborescence contains $r''$. Note that $s$ is distinct from $r$ and $r'$ since it can reach $r''$ while the other two cannot. This root $s$ can clearly connect to $r$ and $r'$, meaning that there is a directed path connecting two distinct roots of $F$. 
\end{proof}

\begin{figure}
    \centering
    \begin{subfigure}{0.4\textwidth}
    \includegraphics[width=0.95\linewidth]{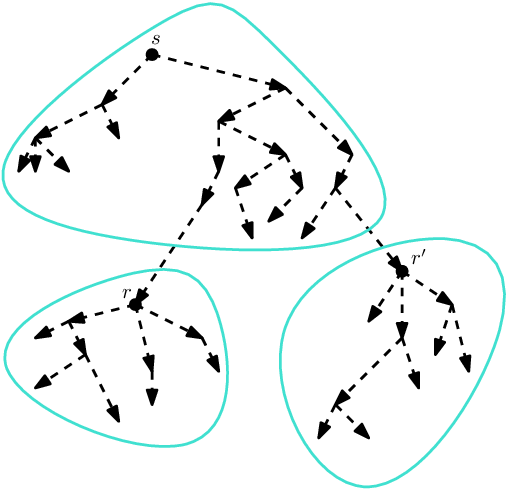} 
    \caption{}
    
    \end{subfigure}
    \begin{subfigure}{0.4\textwidth}
    \includegraphics[width=0.95\linewidth]{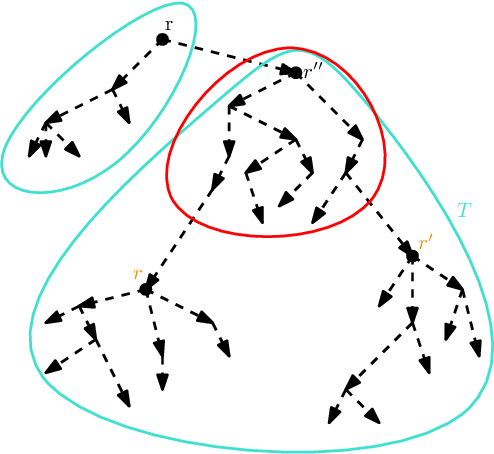} 
    \caption{}
    
    \end{subfigure}
    \caption{Illustration of the second direction of proof of \Cref{onlyroots}. Turquoise: arborescences. Red: strongly connected component of $r''$. Arcs dashed to signal that there could be more arcs. Root labels of $s, r, r'$ and $r''$ correspond to roots in proof. Roots $r$ and $r'$ on right in orange since they are no longer roots.}
    \label{fig:root to root image}
\end{figure}

\begin{lemma}\label{lem:inCompInArb}
    Given a digraph $G$ and a maximum arborescence forest $F$ with a root $r$, if $v$ can reach $r$ in $G$ then $v$ is in the arborescence of $r$ in $F$.
\end{lemma}
\begin{proof}
    Suppose $F$ is a maximum arborescence forest and $v$ is a vertex that can reach a root $r$, however it is not in the arborescence of $r$. Call $r'$ the root of the arborescence that contains $v$. By definition of an arborescence, there is a directed path from $r'$ to $v$ and by assumption there is a path from $v$ to $r$. By \Cref{onlyroots} we get a contradiction to the assumption that $F$ is a maximum arborescence forest, since there is a path from root $r'$ to root $r$.
\end{proof}

\begin{lemma}\label{lem:feasPathExist}
    Suppose $F$ is a maximum arborescence forest for digraph $G = (V, A)$. If $G+a$ for some arc $a \not \in A$ contains a path between two roots of $F$, then $G+a$ contains a feasible path for $F$.
\end{lemma}
\begin{proof}
Let $P$ be the assumed path in $G+a$ from root $r$ to root $r'$ and let $a = (u, v)$. Let $w$ be the last vertex in $P$ which is not contained in the arborescence of $r'$ and let $P_V$ be the maximal suffix of $P$ which does not contain $w$. By definition, there exists some root $r'' \neq r'$ with some arborescence $T''$ of $F$ which has some path $P_A \subseteq T''$ to $w$. Take as the feasible path $P_A \oplus P_V$.







\end{proof}

\begin{lemma}\label{lem:correctness}
Let $F^{(0)}, F^{(1)}, \ldots, F^{(m)}$ be the output of \Cref{alg:main} on graphs $G^{(0)}, G^{(1)}, \ldots, G^{(m)}$. Then $F^{(i)}$ is a maximum arborescence forest for $G^{(i)}$ for every $i$.
\end{lemma}
\begin{proof}
Our proof is by induction on $i$, the above lemmas and, in particular, our careful definition of feasible paths (\Cref{dfn:feasiblePath}). $F^{(0)}$ is trivially a maximum arborescence forest for $G^{(0)}$.

Thus, consider $i > 0$. By our inductive hypothesis, we know that $F^{(i-1)}$ is a maximum arborescence forest for $G^{(i-1)}$. We case on whether adding $a_i$ created a path between roots of $F^{(i-1)}$ or not.
\begin{itemize}
    \item If there are no paths between any roots of $F^{(i-1)}$ in $G^{(i)}$ then $F^{(i)} = F^{(i-1)}$ is maximum by \Cref{onlyroots}.
    \item On the other hand, suppose that there is a path between roots $r$ and $r'$ of $F^{(i-1)}$ in $G^{(i)}$. By \Cref{lem:feasPathExist}, $G^{(i)}$ contains a feasible path $P$ of the form $P_A \oplus P_V$ where the edges of $P_A$ are contained in the arborescence of $r$ in $F^{(i-1)}$ and the vertices of $P_V$ are contained in the arborescence of $r'$ in $F^{(i-1)}$. 
    Furthermore, observe that by the definition of $\upd$ and \Cref{lem:updatePresForest}, we have that $F^{(i)}$ is an arborescence forest with roots $R^{(i)} = R^{(i-1)} - r'$ where $R^{(i-1)}$ are the roots of $F^{(i-1)}$.
    
    We claim that $F^{(i)}$ therefore has no paths between roots in $G^{(i)}$ and so $F^{(i)}$ is maximum in $G^{(i)}$ by \Cref{onlyroots}. Suppose for the sake of contradiction that such a path existed from root $r_0$ to root $r_1$ where $r_0, r_1 \in R^{(i)}$. Since $F^{(i-1)}$ is maximum in $G^{(i-1)}$, it follows by \Cref{onlyroots} that such a path must use $a_i$ (since otherwise we would have a path between two distinct roots of $F^{(i-1)}$ in $G^{(i-1)}$, contradicting \Cref{onlyroots}). Since $r' \not \in R^{(i)}$ it must be the case that $r_1 \neq r'$. However, letting $v$ be the head of $a_i$, we have that $v$ must be reachable from $r'$ in $F^{(i-1)}$ by \Cref{lem:inCompInArb} and the fact that $v$ can reach $r'$ in $F^{(i-1)}$ (by following a suffix of $P$). It follows that $r_1$ has a path to $r' \neq r_1$ in $G^{(i-1)}$, contradicting \Cref{onlyroots}. 
\end{itemize}
\end{proof}

\subsection{Recourse of Algorithm}
\label{subsection: theoretical upper bounds}
In this section we bound the expected recourse of \Cref{alg:main} as $O(m \cdot \log ^ 2 n)$. In particular, we prove the following.

\begin{lemma}
    \label{lem:recourse}
    The expected recourse of \Cref{alg:main} where arcs are inserted uniformly at random (without replacement) is $O(m \cdot \log^2 n)$.
\end{lemma}
\noindent For simplicity of presentation, we will assume in our analysis that $n \log n$ arcs are inserted; it is clear from the analysis that terminating the analysis with $m < n \log n$ still yields the same recourse.

We will make use of the following notion of a component. 
Below, $A(S)$ gives all arcs with both endpoints internal to $S\subseteq V$.

The main fact we will use when bounding recourse is that by our careful choice of paths along which we update---namely feasible paths (\Cref{dfn:feasiblePath})---we only pay recourse when we add an arc incident to a root's in-component and when we do so we pay at most the size of that root's arborescence. 


\begin{lemma}\label{lem:recHelper}
    Suppose \Cref{alg:main} performs an update in iteration $i$, namely $F^{(i)} = \upd(F^{(i-1)}, P)$ where $P$ is a feasible path from root $r$ of $F^{(i-1)}$ to root $r'$ of $F^{(i-1)}$. Let $a_i = (u,v)$ and let $T'$ be the arborescence of $r'$ in $F^{(i-1)}$. Then $v$ is in the in-component of $r'$ in $G^{(i-1)}$ and our recourse is at most
    \begin{align*}
        |F^{(i-1)} \setminus F^{(i)}| \leq  |T'|.
    \end{align*}
\end{lemma}

\begin{proof}
    The fact that $v$ is in the in-component of $r'$ in $G^{(i-1)}$ is trivial.


    To bound the recourse in terms of $|T'|$, observe that the only arcs we delete are parent arcs of vertices in $T'$ which trivially gives $|F^{(i-1)} \setminus F^{(i)}| \leq |T'|$.
\end{proof}

We will also make use of the following lemma which will allow us to argue that after a root $r$ merges with the arborescence of another root $r'$, the in-component of the resulting arborescence is just that of $r$.
\begin{lemma}\label{lem:presInComp}
    Suppose \Cref{alg:main} performs a merge in the $i$th iteration; that is, there is some feasible path from root $r$ to root $r'$ of $F^{(i-1)}$ in $G^{(i)}$. Then the in-component of $r$ is the same in $G^{(i-1)}$ and $G^{(i)}$.
\end{lemma}
\begin{proof}
        In particular, let $G = G^{(i-1)}$ be the graph at this point and let $a = (u,v) = a_i$ be the randomly inserted arc so that $r$ has a path to $r'$ in $G+a$ using $a$. Now suppose for the sake of contradiction that the in-component of $r$ is strictly larger in $G+a$. It follows that there is some $w$ which was not in the in-component of $r$ which now is and this, in turn, can only happen if $w$ has a path to $r$ using $a$. However, it then follows that $v$ has a path to $r$ in $G$ which by \Cref{lem:inCompInArb} means that $v$ is in the arborescence of $r$ (before updating). But, since $v$ can reach $r'$ without using $a$ in $G$ then so too can $r$ and so we have a path from $r$ to $r'$ in $G$ which contradicts \Cref{onlyroots} and \Cref{lem:correctness}.
\end{proof}


Our general approach will be to break our analysis into the first, roughly, $2n$ arcs and then all remaining arcs and appeal to \Cref{lem:recHelper}. More specifically, we consider the following reformulation of uniformly random arc arrivals. Assign to each directed pair of vertices $(u,v) \in V \times V$ a value $\rho_{uv}$ drawn independently and uniformly at random from $[0,1]$. Let $a_1, a_2, a_3, \ldots, a_m$ be the first $m$ such pairs after we sort by $\rho$ values. Throughout the rest of this section we will refer to $\rho_a$ as the value of arc $a$. It is easy to verify that these first $m$ arcs are distributed identically to the $m$ arcs sampled by our uniform random arc arrival model.

What do we gain from such a reformulation? Observe that if we set a fixed threshold $p \in [0,1]$ and define $G_{p}$ to consist of all arcs $a$ which satisfy $\rho_a \leq p$ we have an instance of $D(n,p)$, namely the random digraph on $n$ nodes where each arc appears independently at random with probability $p$. There is a wealth of literature on $D(n,p)$ and so by appropriately setting this threshold $p$ we can hope to leverage some of these results to help bound our recourse. In what follows, we say that the recourse contributed by arc $a_i$ is just $F^{(i-1)} \setminus F^{(i)}$.

\subsubsection{Arcs with Value under 2/n}

We start by bounding the recourse on, roughly, the first $2n$ arcs.

By \Cref{lem:recHelper}, the only arborescences that are a problem are those with super-constant size. We prove this theorem by arguing that all arborescences of non-constant size will get merged at most $\log n$ times (in the regime to which this theorem belongs, i.e. during arc insertion of arcs with value under $2/n$). We do this by showing that after $O(\log n )$ merges it is expected that the root of an arborescence is a vertex whose strong component is a single vertex (i.e.\ the vertex itself). The following result summarizes this.

\begin{restatable}{lemma}{onlogn}\label{theorem: O(nlog(n))}
    All arcs with a value under $2/n$ contribute a total recourse of $O(n \log n)$ with high probability.
\end{restatable}
\begin{proof}
    Let $G_p$ be the graph consisting of all arcs whose value is at most $p$. Notice that, by standard arguments, the number of isolated vertices\footnote{An isolated vertex is one with no incident arcs.} in $G_{2/n}$ is $\Omega(n)$ with high probability (for a suitable small hidden constant). The number of isolated vertices is also clearly monotonically decreasing as more arcs are added and so it follows that $G^{(i)}$ consists of at least $\Omega(n)$ isolated vertices for every $i$ such that $a^{(i)}$ has value at most $2/n$. Each such isolated vertex is clearly a root in any arborescence forest.

    We will consider a vertex $v$ and analyze the number of times this vertex must delete its parent arc. By \Cref{lem:recHelper}, $v$ must only delete its parent arc when the randomly inserted arc is incident to the in-component of the root of the arborescence containing $v$. However, notice that by our assumption on the number of isolated vertices being $\Omega(n)$, we have that such an arc originates in an isolated vertex with constant probability. It follows by \Cref{lem:presInComp} that each time $v$ deletes its parent arc, with constant probability it is afterwards in an arborescence with a root whose in-component is just that root itself. Then, with constant probability for the rest of the phase no arcs will ever be incident to the in-component (and therefore root) of $v$. Thus, we have with high probability that $v$ deletes its parent arc at most $O(\log n)$-many times. It follows that the total recourse across all vertices over this entire phase is $O(n \log n)$ with high probability.
    \end{proof}

\subsubsection{Arcs with Value Over 2/n}
Having bounded the recourse on the first (approximately) $2n$ arcs, we now bound the recourse on all subsequent arcs. Namely, we show that each arc insertion, where the arc has value at least $c/n$ for constant $c > 1$, gives recourse $O(\log^2 n)$ on average. Note that this regime overlaps with the previous regime but the previous regime for $c > 1$ always contains arcs whose values are not in this regime and so the bounds of the previous section are indeed necessary.

We will make use of a well-known fact showing that $D(n,p)$ is strongly connected whenever $p$ is sufficiently large. In particular, this means that we only have to reason about recourse until this strong connectivity threshold is met (at which point the entire graph is spanned by a single arborescence).


\begin{lemma}[\cite{frieze_karoński_2015}, Theorem 13.9]\label{lem:strongConn}
Let $p = 2\log n / n$. Then $D(n,p)$ is strongly connected with high probability.
\end{lemma}

Next, we adapt a known technique to bound the size of components in $D(n,p)$ by using arborescences as witnesses. See \cite{https://doi.org/10.1002/rsa.3240010106} and also Lemma 13.3 of \cite{frieze_karoński_2015}.
\begin{restatable}{lemma}{atMostLogHelp}\label{lemma: at most log n helper}
    Let $p = c/n$ for any $c > 1$. Then with
    high probability neither the in-component nor out-component of any vertex in $D(n,p)$ has size $s$ for $s \in [10 \log n, \beta n]$ where $\beta = 1 - \frac{1.1}{c} - \frac{\ln c}{c}$.
\end{restatable}
\begin{proof}
    Consider a vertex $v$. Note that $v$'s out-component has size $s$, only if it is the root of (at least one) out arborescence $T$\footnote{This does not refer to an arborescences from the arborescence forest but just some arbitrary arborescence rooted at $v$.} of size $s$ with no arcs going from $T$ to $V\setminus T$. 

    The expected number of such arborescences is at most
    \begin{align*}
        s \binom{n}{s}s^{s-2} \left( \frac{c}{n} \right)^{s-1} \left( 1- \frac{c}{n}\right)^{s(n-s)}
    \end{align*}
    Applying $\binom{n}{s} \leq (\frac{ne}{s})^s$ and $1-x \leq e^{-x}$ and rearranging we have that this is at most
    \begin{align*}
        \frac{n}{sc} \left(e^{1 - c + cs/n + \ln c} \right)^s.
    \end{align*}
    Using the fact that $s \leq \beta n$ for $\beta = 1 - \frac{1.1}{c} - \frac{\ln c}{c}$ we get that the above is at most 
    \begin{align*}
        \frac{n}{sc} \left(e^{-.1s} \right).
    \end{align*}
    Lastly, applying the fact that $s \geq 10 \log n$ we have that the above is at most $1/\poly(n)$.
    The result therefore follows by the first moment method and a union bound over all vertices. The bound for the in-component is symmetric.
    \end{proof}


The above results apply to a single draw of $D(n,p)$, but we would like it to apply to all of our relevant $G^{(i)}$. We can achieve this by subdividing our value interval and appropriately union bounding over these intervals. Namely, we show the following.

\begin{restatable}{lemma}{atMostLogn}\label{lemma: at most log n}
    With high probability for any $G^{(i)}$ where $a_i$ has a value larger than $1.6/n$ we have that $G^{(i)}$ contains no vertices with an in-component of size in $[10 \log n, \alpha n]$ for some constant $\alpha > 0$.
\end{restatable}
\begin{proof}
    Let $G_p$ be the random digraph gotten by taking all arcs whose value is at most $p$.

    We subdivide the interval $[1.6/n,1]$ into $n^{10}$ equally-sized contiguous intervals. We can loosely upper bound the probability that a given interval has at least two arcs in it as
    \begin{align*}
        \frac{1}{n^{10}} \cdot \frac{1}{n^{10}} \cdot \binom{n^2}{2} \leq O\left(\frac{1}{n^{16}} \right).
    \end{align*}
    By a union bound over all intervals we have that with high probability no interval has more than one arc in it.
    
    
    

    Next, by \Cref{lemma: at most log n helper} we have that with high probability the graph corresponding to the interval which corresponds to $p = c/n$ for $c \geq 1.6$ has no in-component of size in $[10 \log n, \beta n]$ where $\beta = 1 - \frac{1.1}{c} - \frac{\ln c}{c} \leq 1 - \frac{1.1}{1.6} - \frac{\ln 1.6}{1.6} \leq \alpha$  for some small $\alpha > 0$ independent of $c$. Union bounding over all such intervals gives our result.
    \end{proof}

\begin{restatable}{lemma}{onlogsquaredn}\label{theorem: supercritical regime recourse}
    The expected total recourse from arcs with value larger than $1.6/n$ is $O(n \log^2 n)$.
\end{restatable}
\begin{proof}
First, note that by \Cref{lem:strongConn}, by the time we consider arcs of value $2\log n/ n$ our graph is strongly connected with high probability. It follows by the correctness of our algorithm (\Cref{lem:correctness}) that any time our algorithm performs an update, our graph has at most $5 n \log n$ many arcs with high probability. We condition on this.

In order to now bound the recourse we will fix a vertex $v$ and reason about when we have to delete $v$'s parent edges. In particular, by \Cref{lem:recHelper} we only must delete $v$'s parent arc when we merge along a feasible path from root $r$ to root $r'$ and $r'$ is the arborescence containing $v$. Call an arborescence in-large if the in-component of its root has size at least $\alpha n$ where $\alpha > 0$ is the constant in \Cref{lemma: at most log n}. Call it in-small if the in-component of its root has size at most $10 \log n$. By \Cref{lemma: at most log n} we know that we never have any in-components of size in $[10 \log n, \alpha n]$ with high probability and so these two cases partition all arborescences over the course of our algorithm. We case on what kind of arborescence the arborescence containing $r$ is in such merges.  In particular, we will bucket into a consecutive $n$ rounds of arc insertions and reason about the expected number of times $v$ must delete its parent arc.
\begin{itemize}
    \item First, we bound the number of times $v$ is merged into an in-large arborescence (meaning that after the merge $v$ belongs to an in-large arborescence). By \Cref{lem:inCompInArb} and \Cref{lem:correctness} we know that any in-component of a root is always fully contained in its arborescence and so any in-large arborescence also has cardinality at least $\alpha n$. Furthermore, when we use a feasible path $P$ from root $r$ to root $r'$, the arborescence of $r'$ joins that of $r$ (or, more precisely, the vertices of the arborescence of $r'$ become vertices of the arborescence of $r$). It follows that the number of times $v$ can be merged into an in-large arborescence is at most $O(1)$ (deterministically, since there are at most $O(1)$ in-large arborescences) and so the total recourse incurred by $v$ for such merges is at most $O(1)$ over our $n$ rounds.

    \item Next, we bound the number of times $v$ is merged into an in-small arborescence. 

   The resulting arborescence containing $v$ must be in-small by \Cref{lem:presInComp}.

    It follows that after one such merge $v$ is in an in-small arborescence. Furthermore, notice that when $v$ is in an in-small arborescence it must only delete its parent arc when the randomly inserted arc has a head in the in-component of the root of its arborescence (this follows by \Cref{lem:recHelper}). By assumption, this in-component has size at most $10 \log n$. Since we have assumed that the total number of arcs in the graph is never more than $5n \log n$, the probability of such a random arc having a head in this in-component is certainly at most
    \begin{align*}
        \frac{10 n \log n}{n(n-1) - 5n \log n} \leq \frac{20 \log n}{n}.
    \end{align*}

    We define for a vertex $v$ and every arc insertion $i$ a random variable $X_i$ that is $1$ if $v$ is in a in-small arborescence and needs to delete its parent arc after arc $i$ is added and $0$ otherwise. Note that we set $X_i$ to $0$ when $v$ gets merged into an in-large arborescence. Also, observe that the $X_i$ are not independent. However, the probability of each $X_i$ being $1$ is dominated by a Bernoulli and so we can still invoke Chernoff-bound-like results; the rest of this proof formalizes this idea.
    
     Let $Y_1, Y_2, \ldots$ be independent Bernoulli random variables with probability $p = 20 \log(n)/n$ and then let $X$ and $Y$ be the respective sums of all $X_i$ and $Y_i$ in the regime of this lemma (i.e.\ with value between $1.6/n$ and $5n\log n$). For simplicity of notation we assume that $X_1$ is the first arc we regard in this lemma. Given any $j \geq 1$ and any $S \subseteq [j-1]$, notice that we have 
    \begin{align}\label{eq:hda}
        \Pr(X_j = 1 \mid \cap_{i \in S} X_i = 1) \leq \Pr(Y_j = 1).
    \end{align}
    Furthermore, we have $E[Y] = \mu = 20\cdot (5\log^2(n) - 1.6\log(n)) = \Theta(\log^2(n))$. Thus, applying a Chernoff bound gives 
    \begin{align*}
        Pr(Y \ge 2\mu) \le e^{-4/3 \mu} \le e^{-\log^2(n)}.
    \end{align*}

    It follows that the probability that there is a subset $I$ of size at least $2 \mu$ such that $Y_i = 1$ for every $i \in I$ is at most $e^{-\log^2(n)}$ and so by the independence of the $Y_i$s we have
    \begin{align*}
        \sum_{I : |I| \geq 2\mu} \prod_{i \in I} \Pr(Y_i = 1) \leq e^{-\log^2(n)}.
    \end{align*} 
    Next, notice that
    \begin{align*}
        \Pr(X \geq 2 \mu) \leq \sum_{ k \geq 2\mu} \sum_{I : |I| = k} \Pr(\cap_{i \in I} X_i = 1).
    \end{align*}
    By the conditioning fact given by \Cref{eq:hda}, we have that $\Pr(\cap_{i \in I} X_i = 1)$ is at most $\prod_{i \in I} \Pr(Y_i = 1)$ and so combining the above we get
    \begin{align*}
        \Pr(X \geq 2 \mu) \leq \sum_{ k \geq 2\mu} \sum_{I : |I| = k} \prod_{i \in I} \Pr(Y_i = 1) \leq e^{-\log^2(n)}.
    \end{align*}

    Thus, we have that over the course of these $5 n\log(n)$ rounds with high probability $v$ is not merged into an in-small arborescence more than $O(\log ^ 2 n)$ times. 
\end{itemize}
Thus, a given vertex contributes, in expectation, $O(\log ^ 2 n)$ to recourse over $n \log(n)$ rounds of arc insertions and so the total expected recourse over these $n \log(n)$ rounds is $O(n \log ^ 2 n)$.
\end{proof}


\Cref{lem:recourse} is immediate by combining \Cref{theorem: O(nlog(n))} and \Cref{theorem: supercritical regime recourse}.

\subsection{Putting Correctness and Recourse Bounds Together}
Our proof of \Cref{thm:upper} is immediate by combining the above recourse and correctness bounds.
\upperTheorem*
\begin{proof}
    We use \Cref{alg:main}. Returned solutions are maximum arborescence forests by \Cref{lem:correctness}. The expected recourse is $O(m \cdot \log ^ 2 n)$ by \Cref{lem:recourse}. The runtime is trivial.
\end{proof}

\section{Conclusion}

In this work, we gave the first dynamic algorithms for an arborescence problem with arc insertions---namely, maximum arborescence forest---with amortized expected recourse $O(\log ^ 2 n)$ when arcs are randomly inserted. Likewise, we observed that random arc insertions are, in some sense, necessary for such results, since adversarial insertions can result in amortized recourse $\Omega(n)$.

\newpage
\printbibliography

\newpage
\appendix
\section{Adversarial Case for Arborescences}
\label{subsection: adversarial case}
In this section we show that for the min-cost (spanning) arborescence problem $\Omega(n^2)$ recourse is necessary after $O(n)$ arc insertions. This matches a trivial upper bound of $O(n^2)$. This problem is defined as follows.
\begin{definition}[Min-Cost (Spanning) Arborescence]
    Given digraph $G = (V,A)$, non-negative edge weights $w$ and root $r \in V$ our goal is to compute an $r$-out arborescence $T \subseteq A$ containing all vertices of $V$ of minimum cost.
\end{definition}

The incremental version of min-cost arborescence we consider is as follows.
\begin{definition}[Incremental Min-Cost Arborescence]\label{dfn:minCostIncArb}
    In the incremental version of min-cost arborescence, we start with an edge-weighted graph $G = (V, A)$ and a root $r$, some min-cost $r$-out arborescence $T \subseteq G$ and then weighted arcs $a^{(1)}, a^{(2)}, \ldots \not \in A$ arrive and our goal is to compute a sequence of $r$-out arborescences $T = T^{(0)}, T^{(1)}, T^{(2)}, \ldots$ where $T^{(i)}$ is a min-cost arborescence for $G^{(i)} := (V, A \cup \bigcup_{j \leq i} a^{(j)})$ while minimizing the recourse
    \begin{align*}
        \sum_i |T^{(i)} \setminus T^{(i-1)}|.
    \end{align*}
\end{definition}

It is clear that there are some graphs where even in this setting the adversary has no power, for example take an arborescence where all arcs have weight 0. Clearly this tree is optimal no matter what arcs get added and so no change can be forced. However, almost just as easily one can see that an adversary can add $\Theta(n) $ ($n := |V|$) arcs and with every arc the arborescence changes entirely. More formally restated this gives us the following theorem.

\begin{theorem}
    \label{theorem: adversary0or1}
    There is an instance of incremental min-cost arborescence (\Cref{dfn:minCostIncArb}) where all arc weights are $0$ or $1$ with $O(n)$ arc insertions that requires recourse $\Omega(n^2)$.
\end{theorem}

\begin{proof}
    Figure \ref{0or1weight} shows the graph. The idea is to substitute weight $1$ arcs with $0$ weight arcs in a way that forces the (unique) min-cost arborescence to flip all arcs on a path consisting of $\Theta(n)$-many arcs. The three parts of the ``triangle" should all have length $\Theta(n)$. The bottom is necessary for having $\Theta(n)$ arc deletions per improvement where as the sides are used for enforcing $\Theta(n)$ changes. Note that it is necessary for the new arcs on the sides to go up the triangle instead of downwards as otherwise the direction on the bottom would not be forced to change. 
\end{proof}

\begin{figure}[h]
    \centering
    \begin{subfigure}{0.3\textwidth}
    \includegraphics[width=0.95\linewidth]{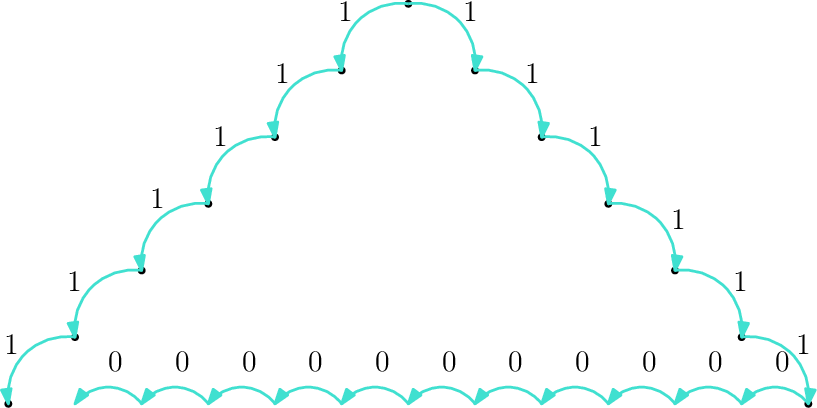} 
    \caption{}
    \label{manychanges1}
    \end{subfigure}
    \begin{subfigure}{0.3\textwidth}
    \includegraphics[width=0.95\linewidth]{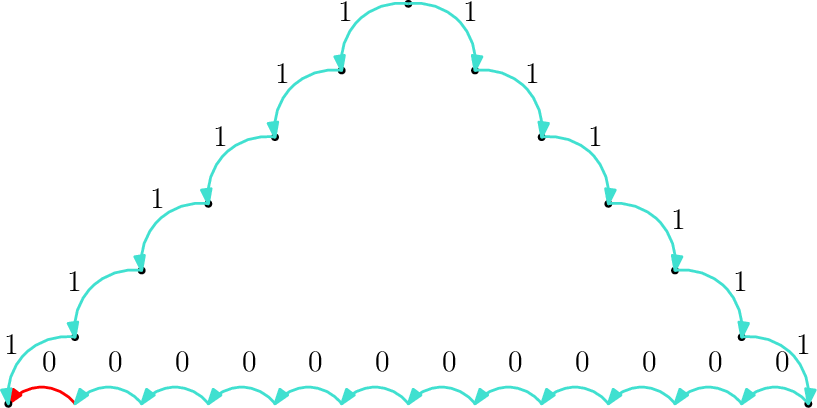} 
    \caption{}
    \label{manychanges2}
    \end{subfigure}
    \begin{subfigure}{0.3\textwidth}
    \includegraphics[width=0.95\linewidth]{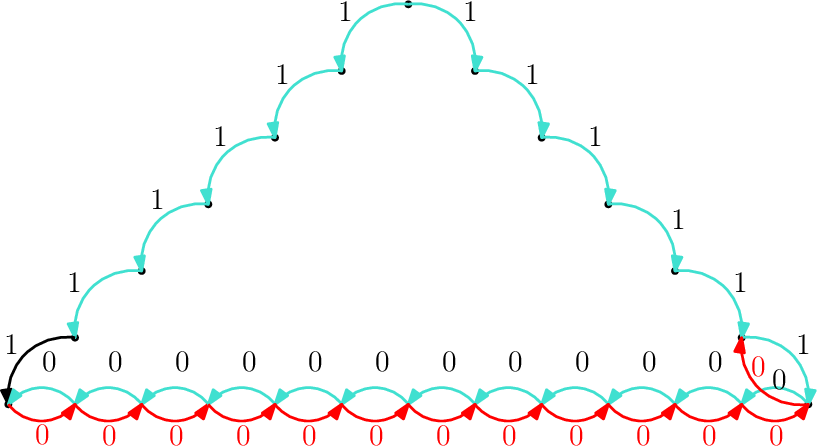} 
    \caption{}
    \label{manychanges3}
    \end{subfigure}
    \begin{subfigure}{0.3\textwidth}
    \includegraphics[width=0.95\linewidth]{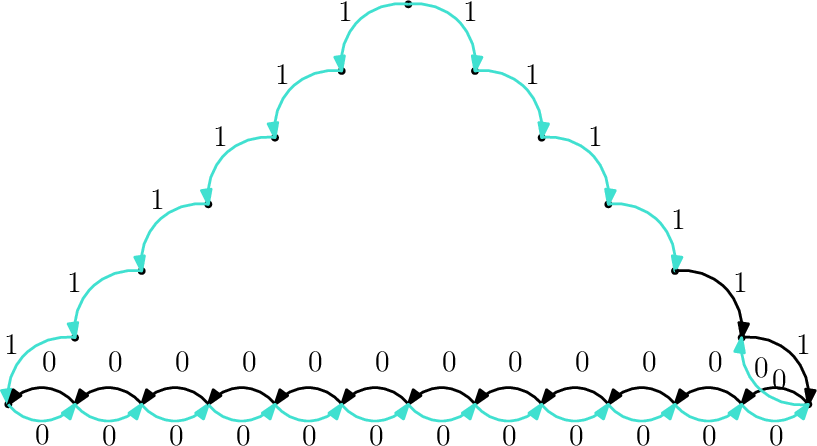} 
    \caption{}
    \label{manychanges4}
    \end{subfigure}
    \begin{subfigure}{0.3\textwidth}
    \includegraphics[width=0.95\linewidth]{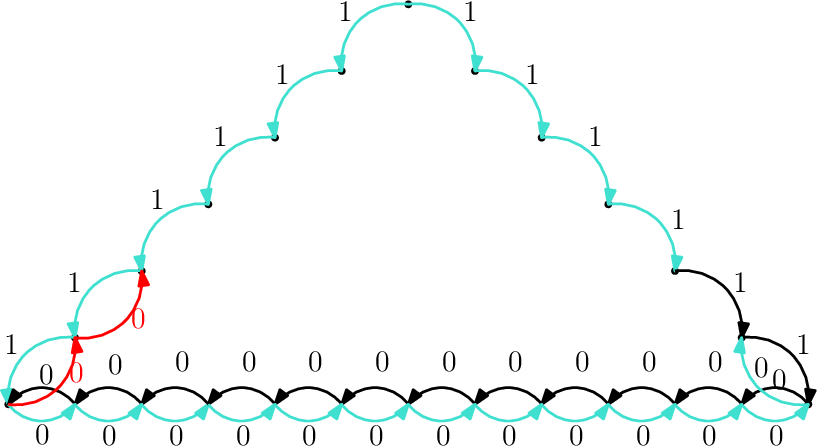} 
    \caption{}
    \label{manychanges5}
    \end{subfigure}
    \begin{subfigure}{0.3\textwidth}
    \includegraphics[width=0.95\linewidth]{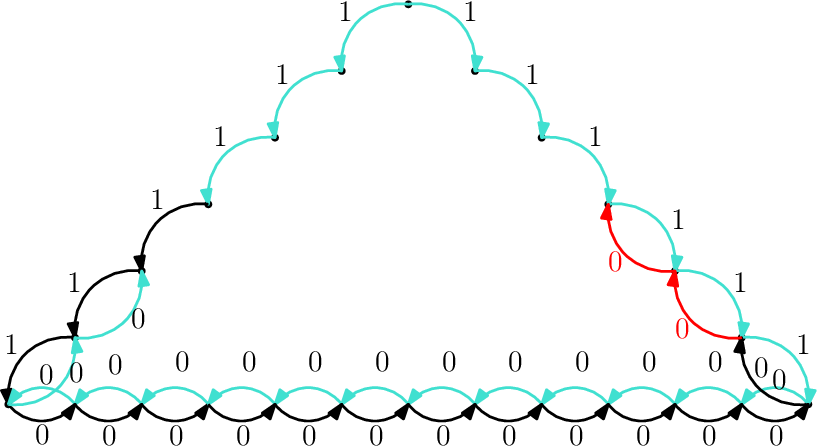} 
    \caption{}
    \label{manychanges6}
    \end{subfigure}
    
    \caption{Red arcs denote arcs that will be added in the next step(s), the turquoise arcs form the optimal arborescence and black arcs are simply currently not in use. The adversary completes and bi-directs the bottom of the triangle and introduces a weight $0$ arc on the right side of the triangle which flips all arcs from the bottom. This is then repeated by adding two $0$ cost arcs on the left side which again flips all bottom arcs for the optimal arborescence. After the bottom arcs have been added, every 2nd arc will flip all bottom arcs which comprise of $n/3$ of all vertices, so a linear number of arcs get swapped every two arcs which gives $O(n)$ arcs added and $\Omega(n^2)$ changes.}
    \label{0or1weight}

    \end{figure}

\section{Algorithm for Dynamic Arborescences}
In this section we give an algorithm for computing and updating the minimum arborescence for dynamically changing graphs. The reason this is in the appendix is that it does not really align with the rest of the paper.

The setting still largely remains the same; Single weighted arcs can get added successively, be it adversarially or randomly, and the goal is to maintain a minimum weight arborescence at all times. We now give the algorithm which has two subroutines.  

\begin{algorithm}[H]
\caption{Arborescence Main Algorithm}\label{alg:cap}
\begin{algorithmic}
\Require $O \gets $ arbitrary r-out-arborescence
\For{$v \in V$ bottom up in $O$} 
    \State $(O',c') \gets GROW(v,c,O)$
    \If{O' = O}
        \State $c' \gets c$
    \ElsIf{$O' \ne O$}
        \State $O \gets O'$
        \State \text{restart with } $c_0$ and $O$
    \EndIf
\EndFor

\end{algorithmic}
\end{algorithm}

\begin{algorithm}[H]
\caption{GROW(v,c,O)}\label{alg:cappp}
\begin{algorithmic}
\Require $c' \gets c$, $O' \gets O$
\While{$c'(Par_O(v)) > 0$}
    \State $S \gets u \in O_v$ s.t. $u$ has distance $0$ to $v$ w.r.t. $c'$
    \State $m \gets$ min value arc into $S$ with source in $O_v \cup Par_O(v)$
    \State $\forall a \in \delta^-(S)$: $c'(a) \gets c'(a) - m$
    \State Check if an arc with updated weight $<0$ into $S$ exists
    \If{$\exists a''\in \delta^-(S)$ s.t. $c'(a'') < 0$ then UPDATE($v,a''$)}
    \EndIf
    \State Return($O',c'$)
\EndWhile

\end{algorithmic}
\end{algorithm}

\begin{algorithm}[H]
\caption{UPDATE($v,a''$)}\label{alg:cappppp}
\begin{algorithmic}
\State Let $a'' = (u, u_1)$
\State $P \gets (u, u_1, ..., v)$ where $c'(u_1, u_2) = ... =  c'(u_l, v) = 0$
\State Delete $Par_O(u_1)$,  $Par_O(u_2)$, ...,  $Par_O(v)$ from $O'$
\State Add $P$ to $O'$

\end{algorithmic}
\end{algorithm}

Denote by $\mathcal{L}$ the laminar family of sets $S$ that get found in Algorithm \ref{alg:cappp} (so a set in $\mathcal{L}$ consists of all vertices that reach their root with updated distance $0$). We start by proving the following claim which will be used for proving correctness of Algorithm $\ref{alg:cap}$.
\begin{claim}
    Let $r'$ be some vertex, $T$ the sub-arborescence starting at $r'$ and $S \in \mathcal{L}$ the subset of vertices in $T$ with distance $0$ to $r'$. Denote by $a'$ the parent arc of $r'$ in the arborescence $O$. Then for any arc $a''$ with starting vertex outside of $T$ and end vertex inside of $S$ the following holds:\\
    There exists a set of arcs $O'$ with $n-1$ arcs s.t.

    1) $r'' \rightarrow T$
    
    2) $c(O') = c(O) + c_S(a'') - c_S(a')$
    
    3) $r\rightarrow V$

    where $c_S(a)$ is the updated cost of a generic arc $a$.

\end{claim}

\begin{claimproof}
    We omit showing that the family $\mathcal{L}$ is laminar.\\
    Firstly $S \in \mathcal{L}$ means that all vertices in $S$ have updated distance 0 to the subroot of $S$.

    Now we do induction over $\mathcal{L}$ starting with the smallest sets in $\mathcal{L}$, the leafs of the arborescence.

    Let $v$ be any leaf, the claim holds trivially (just exchange $a'$ and $a''$).

    Now suppose the claim holds for the first $k-1$ ``levels" of the laminar family $\mathcal{L}$ and consider $S \in \mathcal{L}$ on ``level" $k$.

    Let $a''$ be an arc from outside of $T$ into $S$ and let $r''$ be the end vertex of $a''$. Contract all sets in the laminar family on level $k-1$ to one vertex each and then consider the path $P = (r'', x_1, x_2, ..., x_l, r')$ from the new sub-root $r''$ to the old sub-root $r'$ (this path does exist as it existed before contracting the laminar sets). Now delete $par(r''), par(x_1),..., par(x_l), par(r')$ and add $P$ to the arborescence.

    Expand $r'', x_1, ..., x_l, r'$ with induction hypothesis and call the sets $S_{r''},S_{x_i}, S_{r'}$ respectively. This gives new subtrees for level $k-1$, however, these have the same cost as the original trees. Indeed consider any $S_{x_i}$ (for $S_{r'}$ and $S_{r''}$ the proof works the same). By condition 2 from the claim we have $c(O') = c(O) + c_S(\Tilde{a}'') - c_S(\Tilde{a}')$ where $\Tilde{a}''$ is the arc from $S_{x_{i-1}}$ to $S_{x_{i}}$ and $\Tilde{a}'$ is the parent arc of $S_{x_i}$. Clearly $c_S(\Tilde{a}'') = 0$ since it is used in the 0 cost path from $r''$ to $r'$. Denote by $z$ the end vertex of $\Tilde{a}'$. We also have $c_S(\Tilde{a}') = 0$ as the set $S_{x_i}$ had the parent arc $\Tilde{a}'$ set to zero when we looked at vertex $z$ in the previous iteration of the induction (level $k-1$). So this means that the cost inside of $S$ has not changed and we only exchange $a'$ and $a''$ proving condition 2 from the claim. 

    We now prove condition 1. Let $v \in T$ be any vertex. We show that $r''$ can reach $v$. If $v \in S_{x_i}$ for some $i$ then this holds by just following the path from $r''$ to $S_{x_i}$ and then the induction hypothesis for $S_{x_i}$ finishes the path to $v$ (same if $v \in S_{r'}$ or $v \in S_{r''}$). So now for the other case we consider the unique path from $r'$ to $v$ in the original subarborescence (back when $r'$ was the sub-root) and let $w$ be the last vertex in that original path that is in $\bigcup_i S_{x_i} \cup S_{r'} \cup S_{r''}$. Now simply follow the path from $r''$ to $w$ and from $w$ down to $v$ (the path from $w$ to $v$ has not been changed). This shows condition 1.

    Condition 3 of the claim holds since we do not change the arborescence outside of $T$ apart from $a'$ and $a''$. Furthermore, $r$ reaches the starting vertex of $a''$ and by condition 1 it then also reaches all vertices inside of $T$ (i.e. it reaches $r''$ with $a''$ and $r''$ can reach all vertices in $T$).

    The number of arcs does not change since we delete the same number of arcs as we add to the arborescence.
\end{claimproof}

We now prove that the algorithm is correct, which will conclude this section. We are unsure whether the algorithm is different from the standard Chu–Liu/Edmonds' algorithm as we turned our attention more towards branchings.
\begin{theorem}
The algorithm ``Arborescence Main Algorithm" is correct.
\end{theorem}

\begin{claimproof}
    First we show that the algorithm always returns an r-arborescence in every step. We show this by proving that $r$ can reach any vertex and by showing the number of arcs (in the tree) does not increase.

    Let $v$ be the vertex the algorithm is currently considering. If the GROW Algorithm does not change the tree then we have nothing to show. If it changes the tree then it calls UPDATE which is basically the previous claim. However, this happens with an edge coming from outside the subarborescence starting at v. The claim already tells us that $r$ can reach $V \setminus T$ with $T$ being the subarborescence (before the call of UPDATE) starting at $v$. Now denote $r''$ as the new subroot and $a''$ the parent arc of $r''$. We know that the source vertex, call it $w$, of $a''$ is outside of the subtree starting at $r''$ (which is the same set of vertices as $T$) but this tells us that $r$ can reach $w$ and with $a''$ it can reach $r''$ which, again by the claim, can reach all vertices inside of $T$. 

    Furthermore, the claim tells us that the number of arcs used in the tree does not change so we are done. \\

    Remains to show optimality. We do this by using strong LP-duality and providing cuts with values that satisfy the constraints of the dual (of r-arborescence primal) and the value of the dual is equal to the value of the arborescence we find (thus showing it is minimal).

    Observe that in the GROW subroutine we find sets $S$ and values $m$ which we subtract from all arcs incoming to $S$. These will precisely be the cuts we use for the dual packing. Observe that if all arcs used in the tree have non negative weight and end up with weight $0$ we both have a feasible solution to the dual (no negative arc weights means that all constraints are satisfied) and the value of the dual is equal to the value of the primal (all weights of arcs in the tree have been reduced to $0$).

    Now clearly no arc can have it's weight reduced to a negative number as whenever that would happen we call the UPDATE subroutine.

\end{claimproof}
\end{document}